\documentclass[authorversion=true]{acmart}
\usepackage{latexsym}
\usepackage{makecell}
\usepackage{amsmath,amsthm,mathtools}
\usepackage{amsfonts}
\usepackage{hyperref}
\usepackage[linesnumbered,lined,ruled,noend,vlined]{algorithm2e}
\usepackage{cleveref}
\usepackage{tabularx,environ}
\usepackage{comment}
\usepackage{url}
\usepackage[full]{complexity}
\usepackage{lscape}
\usepackage{lineno}
\usepackage{color}
\usepackage[normalem]{ulem}
\usepackage{xspace}
\usepackage{thmtools}
\usepackage{thm-restate}
\usepackage{appendix}

 \usepackage[disable]{todonotes}
\newcommand{\changed}[1]{#1} 

\newcommand{\set}[1]{\{#1\}}

\newcommand{\size}[1]{| #1 |}
\newcommand{\order}[1]{O\hspace{-0.06cm}\left(#1\right)}

\newtheorem{lemma}{\bf Lemma}

\newtheorem{observation}{Observation}
\newtheorem{problem}{\bf Problem}

\Crefname{algocf}{Algorithm}{Algorithms}

\makeatletter
\@ifpackageloaded{algorithm2e}{
\SetKwProg{Fn}{Function}{}{}
\SetKwProg{Procedure}{Procedure}{}{}
\SetKwProg{Subprocedure}{Subprocedure}{}{}
\SetKwComment{tcc}{//}{}
\SetKw{Continue}{continue}
\SetKwFunction{Output}{Output}%
\SetKwFunction{EnumMC}{EMC}
\SetKwInOut{AlgInput}{Input}
\SetKwInOut{AlgOutput}{Output}
\SetKwInOut{AlgPrecondition}{Pre-conditions}
\SetKwInOut{AlgInvariant}{Invariants}
}{}
\makeatother

\newcommand{\ms}[1]{{\ifthenelse{\equal{#1}{}}{{ ms}}{{ ms}\xspace\left(#1\right)}}} 
\newcommand{\li}[1]{{\ifthenelse{\equal{#1}{}}{\ell}{\ell\xspace\left(#1\right)}}} 
\newcommand{\posms}[1]{{\ifthenelse{\equal{#1}{}}{{ p}}{{ p}\xspace\left(#1\right)}}} 

\newcommand{\comp}[1]{{\ifthenelse{\equal{#1}{}}{{\tt \mu}}{{\tt \mu}\xspace\left(#1\right)}}}
\newcommand{\mcc}[2]{{\ifthenelse{\equal{#1}{}}{{\tt mcc}}{{\tt mcc}\xspace\left(#1, #2\right)}}}
\newcommand{\minimize}[1]{{\ifthenelse{\equal{#1}{}}{{\Gamma}}{{\Gamma}\xspace\left(#1\right)}}}

\newcommand{\MCS}{\textsc{MCS Enumeration}}

\newcommand{\MCSAS}{\textsc{Another MCS}\xspace}
\newcommand{\MIS}{\textsc{MIS Enumeration}\xspace}

\newcommand{\MISC}{\textsc{MIS Counting}\xspace}
\newcommand{\MCSC}{\textsc{MCS Counting}\xspace}

\newcommand{\CAT}{\ensuremath{\mathsf{CAT}}\xspace}

\newcommand{\mS}{\ensuremath{\mathcal{S}}\xspace}
\newcommand{\isize}{\ensuremath{\|\mathcal{S}\|}\xspace}

\usepackage{tikz}
\usetikzlibrary{calc}
\usetikzlibrary{decorations.pathreplacing, patterns, shapes.geometric}
\tikzstyle{vertex}=[circle, draw, inner sep=0pt, minimum size=6pt]

\setcopyright{cc}
\setcctype{by}
\acmJournal{PACMMOD}
\acmYear{2025} \acmVolume{3} \acmNumber{2 (PODS)} \acmArticle{115} \acmMonth{5} \acmPrice{}
\acmDOI{10.1145/3725252}

\title{The Complexity of Maximal Common Subsequence Enumeration}

\titlenote{This is the AAM version of the article accepted for publication, after peer review, at PODS 2025 (doi: \url{https://doi.org/10.1145/3725252})}

\author{Giovanni Buzzega}
\orcid{0000-0002-4117-5312}
\affiliation{
\institution{University of Pisa}
\city{Pisa}
\country{Italy}}
\email{giovanni.buzzega@phd.unipi.it}

\author{Alessio Conte}
\orcid{0000-0003-0770-2235}
\affiliation{
\institution{University of Pisa}
\city{Pisa}
\country{Italy}}
\email{alessio.conte@unipi.it}
 
\author{Yasuaki Kobayashi}
\orcid{0000-0003-3244-6915}
\affiliation{
\institution{Hokkaido University}
\city{Sapporo}
\country{Japan}}
\email{koba@ist.hokudai.ac.jp}       
        
\author{Kazuhiro Kurita}
\orcid{0000-0002-7638-3322}
\affiliation{
\institution{Nagoya University}
\city{Nagoya}
\country{Japan}}
\email{kurita@i.nagoya-u.ac.jp} 
        
\author{Giulia Punzi}
\orcid{0000-0001-8738-1595}
\affiliation{
\institution{University of Pisa}
\city{Pisa}
\country{Italy}}
\email{giulia.punzi@unipi.it}

\begin{abstract}  
Frequent pattern mining is widely used to find ``important'' or ``interesting'' patterns in data.
While it is not easy to mathematically define such patterns, maximal frequent patterns are promising candidates, as frequency is a natural indicator of relevance and maximality helps to summarize the output. 
As such, their mining has been studied on various data types, including itemsets, graphs, and strings. 
The complexity of mining maximal frequent itemsets and subtrees has been thoroughly investigated (e.g., [Boros et al., 2003], [Uno et al., 2004]) in the literature.
On the other hand, while the idea of mining frequent subsequences in sequential data was already introduced in the seminal paper [Agrawal et al., 1995], the complexity of the problem is still open.

In this paper, we investigate the complexity of the maximal common subsequence enumeration problem, which is both an important special case of maximal frequent subsequence mining and a generalization of the classic longest common subsequence (LCS) problem.
We show the hardness of enumerating maximal common subsequences between multiple strings, ruling out the possibility of an \emph{output-polynomial time} enumeration algorithm under $\P \neq \NP$, that is, an algorithm that runs in time ${\rm poly}(|\mathcal I| + N)$, where $|\mathcal I|$ and $N$ are the size of the input and number of output solutions, respectively.
To circumvent this intractability, we also investigate the parameterized complexity of the problem, and show several results when the alphabet size, the number of strings, and the length of a string are taken into account as parameters.
\end{abstract}



\begin{document}
\maketitle
\todo[inline]{Importantly, the page limits for the camera-ready version will be different from the submission version: 18 pages (PACMMOD) + unlimited references + 6 pages for the appendix.}
\section{Introduction}

Frequent pattern mining problems are a central topic in data mining.
Mining algorithms typically aim to find ``interesting'' or ``important'' patterns, but mathematically defining an interesting or important pattern is not easy. 
A reasonable implementation of this concept, that has been widely studied, is using frequency as an indicator of importance, and enumerating all high-frequency patterns.
In particular, structures such as itemsets, graphs, and strings are often sought after~\cite{10.1145/170036.170072,1184038,10.1007/BFb0014140}.

The fundamental problem of \emph{frequent itemset mining}, for example, can be formulated as follows. 
Given a transaction database $\mathcal T = \set{T_1, \ldots, T_k}$ and a threshold $t$, 
enumerate all sets $V$ that are frequent, i.e., contained in at least $t$ transactions in $\mathcal T$.
When the database is a collection of graphs, the problem is called \emph{frequent subgraph mining}~\cite{DBLP:conf/pkdd/HorvathOR13,10.1145/2629550}, and when it is a collection of strings, it is called \emph{frequent subsequence mining}, where a \textit{subsequence} of a string $S$ is a string that can be obtained by deleting arbitrary occurrences of characters of $S$. 
In the frequent subsequence mining problem, each element of a sequence can sometimes be defined as a subset of a given itemset. In this paper, we consider the less general case in which each element of a sequence is a simple character.

Frequency, however, is not the end of the story, as high-frequency patterns are not always all useful. 
For example, a frequent pattern that is simply a subset of another frequent pattern may not be considered useful, as its information is subsumed by the bigger pattern.
Thus, research in the field focused on additional constraints for pruning redundant information, such as closedness and maximality: a frequent pattern is called \emph{maximal} if no other frequent pattern strictly contains it, and it is called \emph{closed} if it is only contained in patterns with lower frequency than itself. Discarding non-maximal or non-closed patterns greatly reduces the number of patterns, with minimal loss of information, but comes at a computational cost, as enumerating just the maximal or closed patterns can be significantly more challenging. 

To analyze this cost, we typically have to look at the \textit{output} size: the number of solutions may be exponential in the size of the input, so we seek algorithms that run, at the very least, in time polynomial in the output size~\cite{Yannakakis81db}. We call these algorithms \emph{output-polynomial}, or \emph{output-sensitive}.


For the itemset mining problem, the output-sensitive enumeration of maximal/closed frequent itemset has been thoroughly studied~\cite{DBLP:journals/amai/BorosGKM03,DBLP:conf/fimi/UnoKA04}. 
For closed frequent itemset mining, Uno et al. ~\cite{DBLP:conf/fimi/UnoKA04} proposed a linear-delay enumeration algorithm, where the delay is the maximum time elapsed between outputs.
In contrast, the enumeration of maximal frequent itemsets is known to be intractable~\cite{DBLP:journals/amai/BorosGKM03}.
More precisely, Boros et al.\ showed that the decision problem of determining, given a set of maximal frequent itemsets, whether there exists another one or not, is \NP-complete.
This hardness result implies that no output-polynomial time algorithm for enumerating maximal frequent itemsets can exist unless $\P = \NP$.

In the case of subgraph mining, computing frequency often leads
to \NP-complete graph problems, e.g., maximum clique, subgraph isomorphism, Hamiltonian path, and so on~\cite{DBLP:conf/pkdd/HorvathOR13,HorvathBR06}.
Moreover, the maximal/closed frequent subgraph mining is intractable even if the graphs are two trees~\cite{10.1145/2629550}.
Thus, the maximal/closed frequent subgraph mining is hard even when frequency can be computed in polynomial time.

As for maximal/closed frequent \textit{subsequence} mining, due to its countless applications in domains such as text analysis or biological data, many practical algorithms have been developed~\cite{DBLP:conf/ijcai/Wang07,Alessio:ACI:2023,Fumarola:KIS:2015,WU2020105812,10.1145/3011141.3011160}. Surprisingly, however, there are still no results on its theoretical complexity. 

In this paper, we seek to characterize this complexity, and to do so we start from the problem of enumerating maximal common subsequences, hereafter \MCS{}. For a set of strings $\mathcal S = \set{S_1, \ldots, S_k}$ and a string $S$,
the \emph{frequency of} $S$ is the number of strings in $\mathcal S$ that contain $S$ as a subsequence. 
\MCS~asks for every string $S$ whose frequency is exactly $k$, i.e., is common to all strings in $\mS$, and that is maximal in the sense that $S$ is not a subsequence of any other common subsequence of $\mS$. 

MCSs are a more general variant of the classic longest common subsequence (LCS) problem, as LCSs are simply MCSs of maximum length, and their complexity was thoroughly investigated~\cite{AbboudLCShardness}.
In turn, MCSs are a special case of the \textit{closed} frequent subsequences, that are all the strings $S$ whose superstrings have lower frequency then $S$, and its complexity remains open.

The main result of this paper is a negative one, as we prove that no efficient enumeration algorithm exists for \MCS~unless $\P = \NP$. Moreover, we investigate restricted instances of the problem, as well as how this complexity affects related problems such as counting, assessment and indexing. A summary of the problems studied and our results is given in Section \ref{section:contribution}.

\subsection{Enumeration, counting, assessment, and indexing} 
Enumeration problems are related to, but still fundamentally different from, counting problems. 
Both problems concern a specific pattern of which we wish to identify occurrences in a given input instance. The counting problem is tasked with simply returning \emph{the number} of such occurrences, while the enumeration problem has the aim of \emph{outputting} all occurrences of the pattern. 
When addressing an enumeration problem, the total number of solutions can be exponentially larger than the input size, and therefore, simply outputting the solution requires exponential time~\cite{DBLP:conf/kdd/Yang04}.
This is also the case for the problem of enumerating MCSs, which is the focus of this paper. 
To address this inherent problem of exponential complexity, we can take two different approaches. 

Indeed, while it is clear that enumeration also provides an answer to the counting problem, there are cases where counting can be performed faster. Even if enumeration intrinsically provides more information than counting, in some data mining applications it is sometimes sufficient to know the number of solutions, instead of finding them all~\cite{punzi2023paths}.
Unfortunately, counting is also often hard to perform efficiently. For instance, both the problem of counting all frequent itemsets~\cite{DBLP:journals/tods/GunopulosKMSTS03} and the aforementioned maximal frequent itemsets/subsequences/subgraphs problems~\cite{DBLP:conf/kdd/Yang04} have been shown to be \#\P-complete, meaning that no polynomial-time counting algorithm exists, unless $\P = \NP$.  
There is also a further intermediate step between counting and enumeration, called assessment: the problem of determining whether there are more solutions than a given integer threshold $z$. 
When $z$ is not too large, this type of problem can be efficiently solved by using an enumeration algorithm. However, there are several problems where assessment is $\NP$-hard even when the threshold $z$ is small~\cite{DBLP:journals/ipl/BrosseDKLUW24,DBLP:journals/dam/BorosM24,DBLP:journals/algorithmica/KhachiyanBEG08}. 
For example, it is known that the problem of determining whether there are more maximal frequent itemsets than a given integer $z$ is \NP-complete~\cite{DBLP:journals/amai/BorosGKM03}.

Another alternative approach to circumvent the exponential explosion of direct enumeration is to compute an index that stores the solutions~\cite{10.1007/978-3-540-68125-0_22,Minato2010,DBLP:conf/wabi/BuzzegaCGP24,10.1145/3011141.3011160,10.1007/BFb0030837}, which can be seen as an ``implicit'' form of enumeration.
This approach aims to construct a small index, defining operations to efficiently handle several useful queries, such as membership queries, random access and (random order) enumeration of solutions \cite{carmeli2022answering}.
Such an index can take the form of one of several data structures, like labeled Directed Acyclic Graphs (DAGs), variants of decision diagrams, and $d$-DNNF~\cite{DBLP:journals/jair/DarwicheM02,info:hdl/2115/58738,SakaiArxiv2023}.
If an enumeration problem admits a small index that can be quickly computed, then the increase in computation time due to output size can be disregarded as we can still perform practically fast enumeration. 
For our problem of \MCS{}, several such indices are known when the input strings are two: two of them of provably polynomial size~\cite{conte2023isaac,SakaiArxiv2023}, and more recently a practically efficient one, with no worst-case guarantees~\cite{DBLP:conf/wabi/BuzzegaCGP24}.  
It is natural to ask whether these results can be generalized to an index for $k$ strings.



\paragraph{\bf The complexity of enumeration algorithms and related problems. }
\label{subsec:output-sensitive}
As mentioned above, the complexity of enumeration is usually expressed with an output-sensitive analysis~\cite{JOHNSON1988119}.
Consider an instance $\mathcal I$ of an enumeration problem, of size $|\mathcal I|$, and let $N$ be the number of solutions to output. 
An enumeration algorithm is called \emph{output-polynomial} if it runs in $\order{\rm{poly}(|\mathcal I| + N)}$ time.
Moreover, an algorithm is called \emph{quasi-polynomial} if it runs in \changed{$\order{(|\mathcal I|)^{\mathrm{poly}(\log (|\mathcal I|))}}$} time, and  \emph{output-quasi-polynomial} if it runs in \changed{$\order{(|\mathcal I| + N)^{\mathrm{poly}(\log (|\mathcal I| + N))}}$} time.  

If an enumeration algorithm runs in $O(N)$ total time, we say that algorithm takes \emph{Constant Amortized Time} (\CAT).
Finally, we sometimes measure the \emph{delay} of an enumeration algorithm, that is the the maximum computation time elapsed between two consecutive outputs.
An algorithm whose delay is $\order{\text{poly}(|\mathcal I|)}$ is called \emph{polynomial-delay}, and the algorithm is clearly output-polynomial since its total running time is $\order{N\cdot \text{poly}(|\mathcal{I}|)}$. 


We also evaluate the efficiency of computing an index, i.e., any compact data structure that allows efficient membership test and enumeration on the solutions.\footnote{Indices can provide more refined operations like random access and random order enumeration \cite{carmeli2022answering}, but at least membership and enumeration should be guaranteed.}
Let $\mathcal D$ be an index that stores all solutions of an instance of an enumeration problem. 
Classic examples are indices that can be represented with edge-labeled graphs, such as tries, labeled DAGs, and SeqBDD/ZDD, where the size $\size{\mathcal D}$ is the sum of the number of vertices and edges. Such indices are particularly interesting when they can be computed in just polynomial time.

Finally, we study and evaluate the complexity of assessment and counting problems. The analogous of \NP-hardness (resp. \NP-completeness) for counting problems is given by \#\P-hardness (resp. \#\P-completeness)~\cite{VALIANT1979189}: a \#\P-hard problem cannot have a polynomial-time counting algorithm, unless $\P = \NP$. 
To be able to gain insight on such hard problems, assessment, which can be seen as a ``partial'' form of counting, is a valid tool. Indeed, even \#\P-complete problems may allow efficient assessment algorithms for polynomial values of $z$~\cite{punzi2023paths}. 
We note that if a counting problem is polynomial-time, then assessment is trivially polynomial-time as well. On the other hand, if the counting problem is \#\P-complete, we cannot perform assessment in time polynomial in the size of the input, as a binary search on the threshold $z$ would allow us to solve the original counting problem in polynomial time as well. 
Thus, an assessment algorithm is often considered efficient if it runs in $\order{\text{poly}(|\mathcal I| + z)}$ time. Since $z$ can be exponential in the size of the input, the strongest hardness results are the ones where we assume $z = O(|\mathcal{I}|)$.

\begin{table}[t]
  \centering
    \begin{tabular}{r|c|c|c|c|c}
      Problem & General & $\size{\Sigma} = 2$ & $k=2$ & $k=O(1)$ & $n=O(1)$ \\
      \hline
      \MCSAS{} & \makecell{\NP-C even \\ if $\size{\mathcal Z} = \order{n}$
      } & -  &\P~\cite{conte2023isaac} &\P &\P  \\
      \hline
      \makecell{\textsc{MCS Assessment} \\ for $z=\mathrm{poly}(n)$}& \makecell{\NP-H even \\ if $z = \order{n}$
      } & -  &\P~\cite{conte2023isaac} & \P  & \P   \\
      \hline
      \textsc{MCS Counting} & \makecell{\#\P-C} & \#\P-C  & \P~\cite{conte2023isaac} & - &\P   \\
      \hline
      \textsc{MCS Enumeration} & No OP alg. & MIS-H  & \CAT~\cite{conte2023isaac} & \P-delay &\P   \\
      \hline
      \textsc{MCS Indexing} & No OP alg. & MIS-H  & \P~\cite{conte2023isaac} & - &\P   \\
      \end{tabular}
  \caption{Summary of our results on the complexity of MCS enumeration and related problems, where $k$ is the number of input strings and $n$ their maximum length. \NP-H, \NP-C and \#\P-C are short for \NP-hard, \NP-complete and \#\P-Complete, respectively. No OP alg.\ means no Output-Polynomial time algorithm exists unless $\P = \NP$. MIS-H means the problem is at least as hard as \MIS{}. The complexity of problems marked with `-' is open. Results for $k=2$ are included for completeness, but follow from~\cite{conte2023isaac}. Results for $n=O(1)$ can be generalized up to values of $n$ polylogarithmic in the input size, but for readability the discussion is only given in Section~\ref{section:parameterized}. }
  \label{tab:neg:results}  
\end{table} 

\subsection{Problem definitions and results}
\label{section:contribution}
In this work, we investigate the complexity of 
\MCS{}, as well as related problems on MCS, with the aim of giving a complete picture of which problems are tractable, and which are not.
In the following, for a given set of strings $\mathcal S$, we denote with $\|\mathcal S \|$ its total size, i.e. the sum of the length of the strings in $\mathcal S$.

Formally, the MCS enumeration problem can be stated as follows:
\begin{problem}[\MCS{}]
    Given a set $\mathcal S$ of strings, the task of \MCS{} is to output all maximal common subsequences of $\mathcal S$.
\end{problem}


As our main result, we show that there is no output-polynomial-time algorithm for \MCS{} in multiple strings, unless $\P = \NP$. 
To show this hardness, we start from the following decision problem:
\begin{problem}[\MCSAS{}]
    Given a set $\mathcal S$ of strings and a set $\mathcal Z$ of maximal common subsequences of $\mathcal S$,
    \MCSAS{} asks whether there exists a maximal common subsequence of $\mathcal S$ 
    that is not contained in $\mathcal Z$.
\end{problem}
In \Cref{section:hardness-another}, we prove the hardness of \MCSAS{} through a reduction from 3-\SAT{}, yielding:
\begin{restatable}{theorem}{hardnessMCSAS}
\label{thm:hardnes:mcsas}
    \MCSAS{} is \NP-complete, even for instances where $|\mathcal{Z}| = O(n)$, where $n$ is the maximum length of an input string. 
\end{restatable}

This type of \emph{another solution} problems (also called finished decision problems~\cite{DBLP:phd/basesearch/Bokler18}, or additional problems~\cite{BABIN2017316}) are frequently used to show the hardness of the corresponding enumeration problems. This implication is folklore~\cite{DBLP:phd/basesearch/Bokler18,DBLP:journals/algorithmica/KhachiyanBEG08,DBLP:journals/dam/BorosM24,DBLP:journals/amai/BorosGKM03,10.1145/2629550};
intuitively, if \MCSAS{} has negative answer (and thus $\mathcal{Z}$ is the whole set of MCS), then any output-sensitive enumeration algorithm for MCS would terminate in time $O(\mathrm{poly}(|\mathcal{Z}|))$. More details concerning this are given in \Cref{sec:relation:hardness}. 
Thus, the \NP-hardness of \MCSAS{} implies the following hardness result for \MCS{}:
\begin{restatable}{corollary}{hardnessMCS}
\label{thm:hardnes:mcs}
    There is no output-polynomial time algorithm for \MCS{}, unless $\P = \NP$.
\end{restatable}

We also focus on the task of creating an index that stores all MCSs, which we call \textsc{MCS Indexing}. Specifically, we are concerned with indexes that allow us to perform efficient membership queries and enumeration for the MCSs. 
Examples of such indexes are tries, labeled DAGs and SeqBDD/ZDD.  The complexity of this problem has already been studied in the literature \cite{conte2023isaac,SakaiArxiv2023} for the case of $k=2$ strings, we here study the general case:
\begin{problem}[\textsc{MCS Indexing}]
    \textsc{MCS Indexing} asks, given any input set of strings $\mathcal S$, to output an index $\mathcal{D}$ in time and space $O(\mathrm{poly}(\isize))$, which stores the set of maximal common subsequences of $\mathcal S$, allowing output-polynomial enumeration and polynomial-time membership testing of the maximal common subsequences in $\mS$. 
\end{problem}
\Cref{thm:hardnes:mcs} immediately implies the hardness of \textsc{MCS Indexing}: if we can construct such an index, then we can enumerate all MCSs in output-polynomial time. Thus, we arrive at the following:
\begin{restatable}{corollary}{hardnessMCSIndex}
\label{thm:hardnes:mcs:index}
    There is no output-polynomial time algorithm for 
    \textsc{MCS Indexing} 
    unless $\P = \NP$.
\end{restatable}

To give a complete picture, we also investigate other problems related to enumeration and mining, namely assessment and {counting}. 
In the MCS setting, the assessment problem is formally stated as:
\begin{problem}[\textsc{MCS Assessment}]
    Given a set $\mathcal S$ of strings and an integer $z$, \textsc{MCS Assessment} asks whether or not the number of maximal common subsequences in $\mathcal S$ is more than $z$.
\end{problem}
For this problem, \Cref{thm:hardnes:mcsas} directly implies that it is \NP-hard to even determine whether the number of MCSs exceeds the maximum length $n$ of the input strings, leading to:
\begin{restatable}{corollary}{hardnessMCSASSESS}
\label{thm:hardnes:mcsassess}
    \textsc{MCS Assessment} is \NP-Hard, even if $z = O(n)$ where $n$ is the maximum length of an input string. 
\end{restatable}

The final problem we address is the problem of MCS counting, formally defined as:
\begin{problem}[\textsc{MCS Counting}]
    Given a set $\mathcal S$ of strings, 
    \textsc{MCS Counting} asks for
    the number of maximal common subsequences in $\mathcal S$.
\end{problem}

In this regard, we dedicate \Cref{section:binary-hardness} to showing a strong direct link between MCSs and maximal independent sets in hypergraphs, in the form of a one-to-one reduction from the problem of enumerating maximal independent sets in hypergraphs (\MIS{}) to \MCS{}.

\begin{restatable}{theorem}{misReduction}
\label{thm:mis-reduction}
    Given a hypergraph $\mathcal H = (V, \mathcal E)$, there exists a set of binary strings $\mathcal S(\mathcal H)$, computable in \changed{time} polynomial in $\|\mathcal{H}\|$, and a bijection between maximal independent sets in $\mathcal H$ and the set $MCS(\mathcal{S}(\mathcal{H})) \setminus \set{w}$, where the string $w$ consists of $\size{V}-1$ repetition of the string ``01''.
\end{restatable}

Notably, the reduction holds even when the strings are constrained to use 
a binary alphabet. \MIS{} (also called the minimal transversal enumeration, the minimal set cover enumeration, and the minimal hitting set enumeration) is a long-standing open problem in enumeration algorithm area~\cite{EITER20082035}. The existence of an output-polynomial enumeration algorithm for \MIS{} has long been questioned, thus we call \textsc{MIS-H} the class of problems that can be reduced to \MIS{}, i.e., that are at least as hard.

The reduction has several implications; first of all, designing an output-polynomial-time enumeration algorithm for \MCS{} is challenging even if the alphabet is binary:
\begin{restatable}{corollary}{hardnessMCSBin}
\label{thm:hardnes:mcsbin}
    If \MCS{} for binary strings can be solved in output-polynomial time, \MIS{} can be solved in output-polynomial time.
\end{restatable}

Similarly as before, the hardness result for indexing follows:
\begin{restatable}{corollary}{hardness}
\label{thm:hardnes:mcs:index-bin}
    If \textsc{MCS Indexing} 
    has an output-polynomial time algorithm,
    then \MIS{} in hypergraphs can be solved in output-polynomial time.
\end{restatable}

Moreover, while the exact complexity of \MIS{} is open, \MISC is known to be \#\P-complete, even in the restricted case of simple graphs (i.e., hypergraphs where all hyperedges have size 2) \cite{doi:10.1137/S0097539797321602}. The above bijection implies that any algorithm for \MCSC can solve \MISC, so we obtain the following:


\begin{restatable}{corollary}{hardnessMCScount}
\label{thm:hardnes:mcscount-bin}
    \MCSC is \#\P-complete, even on binary strings.
\end{restatable}

Given the previous negative results, one is left to wonder: are there tractable instances of \MCS{} and its related problems? As the alphabet size does not seem promising in this regard, we investigate the remaining parameters $n$ (the length of the strings) and $k$ (the number of strings), and finally present two positive results in \Cref{section:parameterized}.

\newcommand{\ml}{\ensuremath{l}\xspace}

As a first positive result, Section~\ref{section:param-string-len} shows that bounding the length $\ml$ of the \textit{shortest} input string already yields an efficient solution to \MCS{}.\footnote{Clearly, $\ml \le n$, so this is strictly stronger than a parametrization in $n$.}

\begin{restatable}{theorem}{paramMinLen}
    \MCS{} can be solved in 
    \FPT~total time with respect to $\ml$. Specifically, when $\ml$ is logarithmic in $\isize$, all MCSs can be enumerated in polynomial time in $n$ and $k$, and when $\ml$ is polylogarithmic in $\isize$, enumeration takes quasi-polynomial time in $n$ and $k$.
\end{restatable}

Secondly, \MCS~has been shown to be output-sensitive for the special case of two input strings, i.e., $k=2$~\cite{DBLP:journals/algorithmica/ConteGPU22}.
In Section~\ref{section:param-num-strings} we generalize the algorithm of~\cite{DBLP:journals/algorithmica/ConteGPU22} showing that output-polynomial enumeration is possible even when $k=O(1)$.

\begin{restatable}{theorem}{paramk}
    \MCS{} for $k\ge 2$ strings of maximum length $n$ can be solved in \XP-delay with respect to $k$. Specifically, MCSs can be enumerated in $O(kn^{2k+1})$ time delay, after $O(k|\Sigma| n^k)$ time preprocessing.
\end{restatable}
\todo[inline]{TODO: Write some concluding remarks to link the positive results with the negative ones?}
We summarize both the negative and the positive results in \Cref{tab:neg:results}.

\subsection{\MCS{} and \textsc{MCS Indexing} are at least as hard as \MCSAS{}}\label{sec:relation:hardness}
In this \changed{section}, we give a brief overview of the relationship between \emph{another solution} problems (also called finished decision problems~\cite{DBLP:phd/basesearch/Bokler18} and additional problems~\cite{BABIN2017316}) and enumeration problems.
Intuitively, to solve an enumeration problem, we can repeatedly solve its \emph{another solution} version until all solutions are found.
In this section we formally show that the \NP-hardness of the another solution problem indeed rules out an output-polynomial time algorithm for the enumeration problem, assuming \P $\neq$ \NP.
Although a similar discussion has already been used in several papers~\cite{DBLP:phd/basesearch/Bokler18,DBLP:journals/algorithmica/KhachiyanBEG08,DBLP:journals/dam/BorosM24,DBLP:journals/amai/BorosGKM03,10.1145/2629550}, we retrace it here in order to make our paper self-contained. 
For the sake of simplicity, here we only focus on \MCS and \MCSAS{}.

\begin{theorem}\label{thm:mcs:mcsas}
    If there exists an output-polynomial time algorithm for \MCS{}, 
    then \MCSAS{} can be solved in polynomial time.
\end{theorem}
\begin{proof}
    Let $(\mathcal S, \mathcal Z)$ be an instance of \MCSAS{} and $\mathcal A$ be an output-polynomial time algorithm for \MCS{}.
    Since $\mathcal A$ runs in output-polynomial time, 
    the running time of $\mathcal A$ is bounded by $(\|\mathcal S\| + \|\mathcal M(\mS)\|)^c$, where $c$ is some constant and
    $\mathcal M(\mathcal S)$ is the set of maximal common subsequences of $\mathcal S$.
    We run $\mathcal A$ for $(\|\mathcal S\|+\|\mathcal Z\|)^c$ steps.
    If $\mathcal A$ terminates, we obtain $\mathcal M(\mathcal S)$, and it is easy to determine whether $\mathcal M(\mathcal S)$ contains a maximal common subsequence not contained in $\mathcal Z$.
    If $\mathcal A$ does not terminate within $(\|\mathcal S\|+\|\mathcal Z\|)^c$ steps, it implies that $\size{\mathcal M(\mathcal S)}$ is larger than $\size{\mathcal Z}$. 
    Therefore, $\mathcal Z \subset \mathcal M(\mS)$ and there is a maximal common subsequence not contained in $\mathcal Z$.
    In both cases, we can solve \MCSAS{} in $(\|\mathcal S\|+\|\mathcal Z\|)^c$ time.
\end{proof}

\section{Preliminaries}
Let $\Sigma$ be an \emph{alphabet}.
An element in $\Sigma^*$ is a \emph{string}.
If $\size{\Sigma} = 2$, an element in $\Sigma^*$ is called a \emph{binary string}.
For a string $X$, we denote the length of $X$ as $\size{X}$.
The string with length zero is called an \emph{empty string}.
We denote the empty string as $\varepsilon$.

For strings $X$ and $Y$, $X\cdot Y$ denotes the concatenation of two strings.
As a shorthand notation, we may use $XY$ instead of $X\cdot Y$.
For a string $W = XY$, $X$ is a \emph{prefix} of $W$.
The $i$-th symbol of a string $X$ is denoted as $X[i]$, where $1 \le i \le \size{X}$.
For a string $X$ and two integers $1 \le i \le j \le \size{X}$, 
$X[i,j]$ denotes the \emph{substring} of $X$ that begins at position $i$ and ends at position $j$.
For a string $X$ and an integer $k$, $X^k$ is the string obtained by repeating $X$ $k$-times.
For two strings $X$ and $Y$, $Y$ is a \emph{subsequence} of $X$ if there is an increasing sequence of integers $(\delta_1, \ldots, \delta_{\size{Y}})$ such that $\delta_{1} < \dots < \delta_{\size{Y}}$, and $X[\delta_1]X[\delta_2]\ldots X[\delta_{\size{Y}}]$ equals $Y$.
Moreover, we call such integer sequence $(\delta_1, \ldots, \delta_{\size{Y}})$ an \emph{arrangement} of $Y$ for $X$.
In addition, we call that $X$ is a \emph{supersequence} of $Y$ or $X$ contains $Y$ as a subsequence.

For a set of $k$ strings $\mathcal S = \set{S_1,\ldots, S_k}$, we denote with $\isize$ the size of $\mS$: $\isize = \sum_{i=1}^k |S_i|$. We say that
a string $X$ is a \emph{common subsequence of $\mathcal S$} if 
for any $1 \le i \le k$, $S_i$ contains $X$ as a subsequence.
If any supersequence of $X$ is not common subsequence of $\mathcal S$, 
$X$ is an \emph{MCS of $\mathcal S$}.

\section{\NP-Hardness of \MCSAS{}}
\label{section:hardness-another}

In this section, we focus on proving \Cref{thm:hardnes:mcsas}. From this, the results of Corollaries~\ref{thm:hardnes:mcs}-\ref{thm:hardnes:mcsassess} will immediately follow. 
%
%
%
To prove the theorem, we give a reduction from 3-\SAT{} to \MCSAS, creating an instance of the latter where set $\mathcal{Z}$ has linear size.

Let $\phi$ be a formula in 3-Conjunctive Normal Form (CNF), i.e., composed of conjunctions of clauses, where each clause is made of a disjuction of exactly three literals.
We denote the total number of variables and the number of clauses as $v$ and $m$, respectively.
Without loss of generality, we assume that no clause contains both $x_i$ and $\bar{x}_i$, and no variable appears in every clause (as otherwise, we can split the problem into two easy 2-\SAT{} instances).
Formally, we consider $\phi = \bigwedge_{i=1}^m C_i$, where $C_i = \ell^{(i)}_1 \lor \ell^{(i)}_2 \lor \ell^{(i)}_3$ with $\ell^{(i)}_1 \in \{x_\alpha, \bar{x}_\alpha\}$, $\ell^{(i)}_2 \in \{x_\beta, \bar{x}_\beta\}$, and $\ell^{(i)}_3 \in \{x_\gamma, \bar{x}_\gamma\}$ for some $1 \le \alpha < \beta < \gamma \le v$. 

We construct the corresponding instance $(\mathcal S(\phi), \mathcal Z(\phi))$ of \MCSAS as follows.
For each clause $C_i$ in $\phi$, 
we define a string $S_i = R^{\alpha-1} \ell^{(i)}_{1} R^{\beta - \alpha} \ell^{(i)}_{2} R^{\gamma - \beta} \ell^{(i)}_{3} R^{v-\gamma}$, where $R = x_{v}\bar{x}_{v} \dots x_1\bar{x}_1$.
Moreover, we define $S_0$ as $x_1\bar{x}_1 \ldots x_{v}\bar{x}_{v}$.
The set of strings $\mathcal S(\phi)$ is defined as $\{S_0, \dots, S_{m}\}$.
For each $1 \le i \le v$, 
we let $Z_i = x_1\bar{x}_1 \ldots x_{i-1}\bar{x}_{i-1} x_{i+1}\bar{x}_{i+1}\ldots x_{v}\bar{x}_{v}$, and define $\mathcal Z(\phi)$ as $\set{Z_1, \ldots, Z_{v}}$.
In other words, $Z_i$ is the string obtained by removing $\bar{x}_{i}x_{i}$ from $S_0$.
It is clear that, for each 3-\SAT{} instance $\phi$, this reduction constructs a corresponding instance of \MCSAS{} in polynomial time. \changed{Furthermore, note that $|\mathcal{Z}(\phi)| = v$ is linear in the size of the input. }
We denote the instance $(\mathcal S(\phi), \mathcal Z(\phi))$ as $\mathit{MCS}(\phi)$.

Let us give an example of our reduction: let $\phi = (x_1 \lor \bar{x}_2 \lor \bar{x}_3) \land (x_2 \lor \bar{x}_3 \lor \bar{x}_4) \land (\bar{x}_1 \lor x_3 \lor x_4)$.
Then, the set $\mathcal S(\phi)$ is defined as follows (parentheses are added for readability, but they are not part of the strings):
\begin{itemize}
    \item[] $S_0 = x_1\bar{x}_1x_2\bar{x}_2x_3\bar{x}_3x_4\bar{x}_4$, 
    \item[] $S_1 = x_1(x_4\bar{x}_4x_3\bar{x}_3x_2\bar{x}_2x_1\bar{x}_1)
\bar{x}_2(x_4\bar{x}_4x_3\bar{x}_3x_2\bar{x}_2x_1\bar{x}_1)x_3(x_4\bar{x}_4x_3\bar{x}_3x_2\bar{x}_2x_1\bar{x}_1)$,
    \item[] $S_2 = (x_4\bar{x}_4x_3\bar{x}_3x_2\bar{x}_2x_1\bar{x}_1)x_2 (x_4\bar{x}_4x_3\bar{x}_3x_2\bar{x}_2x_1\bar{x}_1)\bar{x}_3(x_4\bar{x}_4x_3\bar{x}_3x_2\bar{x}_2x_1\bar{x}_1)\bar{x}_4$, and
    \item[] $S_3 = \bar{x}_1(x_4\bar{x}_4x_3\bar{x}_3x_2\bar{x}_2x_1\bar{x}_1)(x_4\bar{x}_4x_3\bar{x}_3x_2\bar{x}_2x_1\bar{x}_1)x_3(x_4\bar{x}_4x_3\bar{x}_3x_2\bar{x}_2x_1\bar{x}_1)x_4$,
\end{itemize}
and the set $\mathcal Z(\phi)$ is defined as follows:
\begin{itemize}
    \item[] $Z_1 = x_2\bar{x}_2x_3\bar{x}_3x_4\bar{x}_4$, 
    \item[] $Z_2 = x_1\bar{x}_1x_3\bar{x}_3x_4\bar{x}_4$,
    \item[] $Z_3 = x_1\bar{x}_1x_2\bar{x}_2x_4\bar{x}_4$, and
    \item[] $Z_4 = x_1\bar{x}_1x_2\bar{x}_2x_3\bar{x}_3$.
\end{itemize}

Notice that, in this example, $x_1x_2x_3x_4\bar{x}_4$ is an MCS of $\mathcal S(\phi)$, and $\phi$ is satisfied by the assignment $(x_1, x_2, x_3, \bar{x}_4) = (\mathrm{True}, \mathrm{True}, \mathrm{True}, \mathrm{False})$.

We first show that our reduction generates a valid instance of \MCSAS{}.
More precisely, we show that each string in $\mathcal Z(\phi)$ is an MCS of $\mathcal S(\phi)$.
To this end, the following observation is trivial but helpful.

\begin{observation}\label{obs:subseq}
\changed{
    Let $S$ and $T$ be two strings.
    Fix a value $1\le j \le |T|$ and
     let $i$ be an index such that $T[1,j]$ is a subsequence of $S[1,i]$, but $T[1,j+1]$ is not.
    Then, $T$ is a subsequence of $S$ if and only if $T[j+1,\size{T}]$ is a subsequence of $S[i+1,\size{S}]$.
    }
\end{observation}
\begin{proof}
\changed{
    If $T[j+1,\size{T}]$ is a subsequence of $S[i+1, \size{S}]$, and $T[1,j]$  is a subsequence of $S[1,i]$, then
    $T$ is a subsequence of $S$.}
    
    \changed{Let us now show the opposite direction.
    Since $T$ is a subsequence of $S$, 
    it has an arrangement $(\delta_1, \ldots, \delta_{\size{T}})$.
    Since $T[1,j+1]$ is not a subsequence of $S[1,i]$, we have that $\delta_{j+1} > i$.
    This means that $T[j+1,\size{T}]$ is a subsequence of $S[i+1,\size{S}]$ with $(\delta_{j+1},\ldots,\delta_{\size{T}})$ as an arrangement. 
    }
    %
\end{proof}

Now, we are ready to show that our reduction is valid.

\begin{lemma}\label{lem:redundant_mcs}
    Let $\phi = C_1 \land \dots \land C_{m}$ be a 3-\SAT{} formula. 
    Then, each $Z_i \in \mathcal Z(\phi)$ is an MCS of $\mathcal S(\phi)$. 
\end{lemma}
\begin{proof}
    From the construction of $S_0$, $Z_i$ is a subsequence of $S_0$.
    For each $1 \le j \le m$, $S_j$ contains $Z_i$ as a subsequence since $S_j$ contains $R^{v-1}$, where $R = x_{{v}}\bar{x}_{{v}}\ldots x_1\bar{x}_1$.
    
    We show each $Z_i$ is maximal by contradiction.
    Let $Z' \neq Z_i$ be a \changed{common }subsequence of $\changed{\mathcal{S}(\phi)}$ 
    that contains $Z_i$ as a subsequence.
    Since $Z'$ is a subsequence of $S_0$, it is of the form $Z' = x_1\bar{x_1}\ldots x_{i-1}\bar{x}_{i-1} X x_{i + 1}\bar{x}_{i + 1}, \ldots, x_{{v}}\bar{x}_{{v}}$, where $X$ is a subsequence of $x_{i}\bar{x}_{i}$.
    Since each $S_j$ does not contain $S_0$ as a subsequence, $X \neq x_{i}\bar{x}_{i}$.
    Therefore, suppose\changed{, without loss of generality,} that $X = \bar{x}_i$.
    Since there is no variable that appears in every clause, we can also assume that no literal appears in every clause; thus,
    there is a clause $C_j = \ell^{(j)}_1 \lor \ell^{(j)}_2 \lor \ell^{(j)}_3$ such that $\ell^{(j)}_1, \ell^{(j)}_2, \ell^{(j)}_3 \notin \set{x_i, \bar{x}_i}$.
    We show that $Z'$ is not a subsequence of $S_j$.
    Let $\alpha$, $\beta$ and $\gamma$ be the integers with $\alpha < \beta < \gamma$ such that 
    $\ell^{(j)}_1 \in \set{x_\alpha, \bar{x}_\alpha}$,
    $\ell^{(j)}_2 \in \set{x_\beta, \bar{x}_\beta}$, and
    $\ell^{(j)}_3 \in \set{x_\gamma, \bar{x}_\gamma}$.
    Suppose that $i < \alpha$.
    \changed{By \Cref{obs:subseq}, since $Z'[1,2(i-1)+1]=x_1\bar{x_1}\ldots x_{i-1}\bar{x}_{i-1}\bar{x}_i$ is a subsequence of $R^{i}$, but $Z'[1,2(i-1)+2]$ is not, we have that $Z'$ is a subsequence of $S_j$ if and only if $x_{i+1}\bar{x}_{i+1}, \ldots, x_{{v}}\bar{x}_{{v}}$ is a subsequence of 
    $R^{\alpha - i -1}\ell^{(j)}_1R^{\beta-\alpha}\ell^{(j)}_2R^{\gamma - \beta}\ell^{(j)}_3R^{v-\gamma}$.
    Since the number of repetitions of $R$ is $v-i -1$, 
    $R^{\alpha - i-1}\ell^{(j)}_1R^{\beta-\alpha}\ell^{(j)}_2R^{\gamma - \beta}\ell^{(j)}_3R^{v-\gamma}$ does not contain $x_{i+1}\bar{x}_{i+1}, \ldots, x_{{v}}\bar{x}_{{v}}$ as a subsequence.
    }
    Suppose that $\alpha \le i < \beta$.
    By applying the same approach, 
    $Z'$ is a subsequence of $S_j$ if and only if 
    $x_{\alpha}\bar{x}_{\alpha} \ldots Xx_{i+1}\bar{x}_{i+1}\ldots x_{v}\bar x_{v}$ is a subsequence of $\ell^{(j)}_1R^{\beta-\alpha}\ell^{(j)}_2R^{\gamma - \beta}\ell^{(j)}_3R^{v-\gamma}$.
    \changed{Since $\ell^{(j)}_1 \neq x_\alpha\bar x_{\alpha}$, }
    removing $x_\alpha\bar{x}_\alpha$ and $\ell^{(j)}_1$ does not change the subsequence relationship.
    Thereafter, by the same argument as for $i < \alpha$, $Z'$ is not a subsequence of $S_j$.
    From the same discussion in other cases, $Z'$ is not a subsequence of $S_j$, regardless of $i$.
\end{proof}

\Cref{lem:redundant_mcs} proves that 
$\mathit{MCS}(\phi)$ is a valid instance of \MCSAS{}.
We next show that there is a common subsequence of $\mathcal S(\phi)$ that is not included in $\mathcal Z(\phi)$ if and only if $\phi$ is satisfiable\changed{, which needs a preliminary lemma}.

\begin{lemma}\label{lem:sufficient}
    Let $X$ be a subsequence of $S_0$ that contains exactly one of $x_j$ or $\bar{x}_j$ for each $1 \le j \le v$. Let $1\le i \le m$ and $C_i = \ell^{(i)}_1 \lor\ell^{(i)}_2 \lor\ell^{(i)}_3$.
    Then, 
    $X$ is a subsequence of $S_i$ if and only if
    $X$ does not contain $\bar{\ell}^{(i)}_1\bar{\ell}^{(i)}_2\bar{\ell}^{(i)}_3$ as a subsequence.
\end{lemma}
\begin{proof}
    Let $\alpha$, $\beta$, and $\gamma$ be the integers with $\alpha < \beta < \gamma$ that satisfy
    $\ell^{(i)}_1 \in \set{x_\alpha, \bar{x}_\alpha}$,
    $\ell^{(i)}_2 \in \set{x_\beta, \bar{x}_\beta}$, and
    $\ell^{(i)}_3 \in \set{x_\gamma, \bar{x}_\gamma}$.
    Observe that $S_i[1 + 2v(\alpha - 1)] = \ell_1^{(i)}$, $S_i[2 + 2v(\beta - 1)] = \ell_2^{(i)}$, and $S_i[3 + 2v(\gamma - 1)] = \ell_3^{(i)}$. 

    Suppose that $X$ does not contain $\bar{\ell}^{(i)}_1\bar{\ell}^{(i)}_2\bar{\ell}^{(i)}_3$ as a subsequence, that is, 
    $X$ contains one of $\ell^{(i)}_1$, $\ell^{(i)}_2$, or $\ell^{(i)}_3$.
    \changed{Suppose that $X[\omega] = \ell^{(i)}_j$ for some $(j,\omega) \in \{(1,\alpha),(2,\beta),(3,\gamma)\}$.}
    \changed{Since $S_i$ contains $R^{\omega-1}\ell^{(i)}_jR^{v-\omega}$, and $R$ contains the whole alphabet, $X$ is a subsequence of $S_i$. }
    
    We show the opposite direction.
    \changed{Suppose by contradiction that $\bar{\ell}^{(i)}_1\bar{\ell}^{(i)}_2\bar{\ell}^{(i)}_3$ is a subsequence of $X$.
    By \Cref{obs:subseq}, $X$ is a subsequence of $S_i$ if and only if $X[\alpha, \size{X}]$ is a subsequence of $S_i[1+2v(\alpha-1),\size{S_i}]$, as 
    $X[1,\alpha-1]$ is a subsequence of $S_i[1,2v(\alpha-1)] = R^{\alpha - 1}$ but $X[1,\alpha]$ is not.
    As $X[\alpha] \neq S_i[1+2v(\alpha-1)]$ and $X[\beta] \neq S_i[2+2v(\beta-1)]$, we can repeat the same argument: $X$ is a subsequence of $S_i$ if and only if $X[\beta,\size{X}]$ is a subsequence of $S_i[2+2v(\beta-1),\size{S_i}]$ which holds if and only if
    $X[\gamma,\size{X}]$ is a subsequence of $S_i[3+2v(\gamma-1),\size{S_i}]$.
    However, $X[\gamma, \size{X}]$ is not a subsequence of $S_i[3+2v(\gamma-1),\size{S_i}] = \ell^{(i)}_3R^{v-\gamma}$, as $X[\gamma] \neq \ell^{(i)}_3$ and $|X[\gamma, |X|]| = v-\gamma+1$ contains more literals than the $v-\gamma$ repetitions of $R$ in $S_i[3+2v(\gamma-1),\size{S_i}]$, a contradiction. 
%
    }
\end{proof}

\begin{lemma}
    Let $\phi$ be a 3-CNF formula.
    Then, $\phi$ is satisfiable if and only if $\mathcal S(\phi)$ has an MCS that is not contained in 
    $\mathcal Z(\phi)$.
\end{lemma}
\begin{proof}
    We first show that
    if $\phi$ is satisfiable, then $S(\phi)$ has an MCS that is not contained in $\mathcal Z(\phi)$.
    Let $f$ be a satisfying assignment of $\phi$.
    Let $S = s_1s_2, \ldots, s_{v}$ be the string obtained by $s_j = x_j$ if $f(x_j) = \rm{True}$, otherwise, $s_j = \bar{x}_j$.
    Note that $S$ is a subsequence of $S_0$.
    Since $f$ is a satisfying assignment, each $C_i = \ell^{(i)}_1 \lor \ell^{(i)}_2 \lor\ell^{(i)}_3$ has a literal $\ell^{(i)}_j \in \set{x_\alpha, \bar{x}_\alpha}$ that is assigned \rm{True}.
    Thus, $S$ is a subsequence of $S_i$ for each $1\le i \le m$ by matching $s_\alpha$ with $\ell^{(i)}_j$\changed{ and each other literal with one occurrence of $R$}.
    Since $S$ is not a subsequence of any string $\mathcal Z(\phi)$, $\mathcal S(\phi)$ has an MCS that is not contained in $\mathcal Z(\phi)$.

    We next show the opposite direction.
    Suppose that there is a common subsequence $X$ of $\mathcal S(\phi)$ that is not contained in $\mathcal Z(\phi)$.
    Since each string $Z_i$ contains both $x_j$ and $\bar{x}_j$ except for $j = i$, we have
    that $X$ contains at least one of $x_j$ or $\bar{x}_j$ for each $1 \le j \le v$.
    Let $X'$ be a subsequence of $X$ such that, for each $1 \le j \le v$, $X'$ contains \changed{exactly} one of $x_j$ or $\bar{x}_j$.
    From \Cref{lem:sufficient},
    for each $C_i = \ell^{(i)}_1 \lor \ell^{(i)}_2 \lor \ell^{(i)}_3$,
    $X'$ does not contain $\bar{\ell}^{(i)}_1\bar{\ell}^{(i)}_2\bar{\ell}^{(i)}_3$ as a subsequence.
    From $X'$, we obtain an assignment $f$:
    if $X'$ contains $x_j$, $f(x_j) = \mathrm{True}$, and otherwise $f(x_j) = \mathrm{False}$.
    From \Cref{lem:sufficient}, $f$ is a satisfying assignment of $\phi$.
\end{proof}

\changed{Note that \MCSAS is in \NP, as we can verify in polynomial time that a given string is not contained in the input set $\mathcal Z$ and that it is indeed maximal for $\mathcal S$ \cite{hirota2023fast}}.
Overall, the provided reduction proves the \NP-completeness of \MCSAS{}, and we obtain:

\hardnessMCSAS*



As mentioned in \Cref{section:contribution}, \MCSAS{} can be solved output-polynomial time under the assumption that \MCS{} admits an output-polynomial time algorithm.
This together with \Cref{thm:hardnes:mcsas} implies that \MCS{} has no output-polynomial time algorithm unless $\P = \NP$.
Similarly, \textsc{MCS Indexing} also has no output-polynomial time algorithm. 
Finally, the \NP-hardness of \textsc{MCS Assessment} with linear sized $z$ follows as well.
Indeed, for any 3-\SAT{} instance $\phi$ with $v$ literals, we have $|\mathcal Z(\phi)| = v$, and by setting $z = v$, polynomial-time algorithms for \textsc{MCS Assessment} solve \MCSAS{} in polynomial time as well. Thus, Corollaries~\ref{thm:hardnes:mcs}-\ref{thm:hardnes:mcsassess} are proved. 

We believe that these hardness results provide strong evidence for the conjecture that \MCS{} and \textsc{MCS Indexing} do not have an output-quasi-polynomial time algorithm unless $\P = \NP$.
If \MCS{} or \textsc{MCS Indexing} had an output-quasi-polynomial time algorithm, then \NP-complete problem could be solved in quasi-polynomial time.
This would imply that any problem in \NP{} could be solved in quasi-polynomial time.

\section{Hardness of \MCS{} on a binary alphabet}
\label{section:binary-hardness}
The reduction described in the previous section only works for strings on unbounded alphabet size. A natural question that arises is whether the problem becomes easier when the alphabet is small.
In this section, we show that \MCS{}, and other related problems, are still \changed{challenging} 
even when the input strings are binary.
More precisely,
we give a one-to-one reduction from \MIS{} to \MCS{} for binary strings, similar to the one used in~\cite[Proposition 1]{10.1007/978-3-642-31265-6_11}.

Before explaining our reduction, we must first introduce some terms related to hypergraphs. \changed{A hypergraph is a pair $\mathcal H = (V, \mathcal E)$, where $ \mathcal E \subseteq \mathcal{P}(V)$.} For a hypergraph $\mathcal H = (V, \mathcal E)$, a set of vertices $U$ is \emph{independent} if
for any hyperedge $E \in \mathcal E$, $E \not\subseteq U$ holds.
An independent set $U$ is a \emph{maximal independent set} of $\mathcal H$ if any proper subset $W \supset U$ is not independent.
The task of \MIS{} is to output all maximal independent sets in $\mathcal H$.
An output-quasi-polynomial time enumeration of \MIS{} is proposed by Fredman and Khachiyan~\cite{FREDMAN1996618}.
No output-polynomial time algorithm for \MIS{} is known so far, and its existence is considered to be a long-standing open problem in this field.
In what follows, we regard $V$ as the set of integers $\set{1, \ldots, \size{V}}$ to simplify the discussion.

Let $\mathcal H = (V, \mathcal E)$ be a hypergraph.
To simplify our proof, we assume that $\mathcal H$ has at least two vertices.
Furthermore, we can assume that $\mathcal H$ has no independent sets with cardinality $\size{V}-1$.
Indeed, let us add a vertex $v$ that 
is contained in all hyperedges, and let us denote as $\mathcal H'$ the resulting hypergraph.
It is easy to verify that the collection of maximal independent sets containing $v$ in $\mathcal H'$ corresponds to the collection of maximal independent sets in the original hypergraph $\mathcal H$.
Moreover, every maximal independent set that does not contain $v$ in $H'$ contains all the vertices in $V$, meaning that $V$ itself is the only such maximal independent set.
If we have an output-polynomial time algorithm for such restricted hypergraphs, \MIS{} can be solved in output-polynomial time for general hypergraphs.
Therefore, without loss of generality, we can assume that $\mathcal H$ has no independent sets of cardinality $\size{V} - 1$.  \changed{This assumption is equivalent to assuming that there exists no vertex belonging to every hyperedge.}

We now give our reduction from \MIS{} to \MCS{} (see~\Cref{fig:reduction} for an example).
We first define $S_0$ as $(01)^{\size{V}}$.
\changed{Let $E_i =\set{u_1, \ldots, u_h}$ be a hyperedge in which for each $1\le j \le \size{E_i}$, $u_j \in V = \set{1,\dots,n}$. Suppose without loss of generality that $u_j < u_{j+1}$ for $1 \le j < \size{E_i}$. For each hyperedge $E_i$,}
we define string $S_i = T_1\ldots T_{\size{V} + \size{E_i} - 1}$,
where $T_j = 0$ if \changed{$j = u_k+k-1$}
for some $1 \le k \le \size{E_i}$, otherwise $T_j = 01$.
In other words, we start to build $S_i$ from $(01)^{\size{V}-1}$, and for each $u_j \in E_i$, we add a further $0$ before the $u_j$-th occurrence of $01$ of the string, or at the end of the string for $u_j= |V|$. As an example, consider $E_3=\{3,4,5\}$ in the hypergraph of \Cref{fig:reduction}: to build the corresponding $S_3$ we start from $(01)^4 = 01010101$, and we add two zeros before the third and fourth one, and a last zero in the last position of the string, yielding $S_3 = 0101\underline{0}01\underline{0}01\underline{0}$.
A key observation of this construction is that $S_i$ contains $(01)^{j - 1}0(01)^{\size{V} - j}$ as a subsequence for each $j$.
We denote the set of strings $S_0, \ldots, S_{\size{\mathcal E}}$ by $\mathcal S(\mathcal H)$.
This reduction can be done in polynomial time.

\changed{Next, we show that} for a set of binary strings $\mathcal S(\mathcal H)$,
any MCS does not contain $11$ as a substring.
\begin{observation}\label{obs:11}
    Any MCS of $\mathcal S(\mathcal H)$ does not contain $11$ as a substring.
\end{observation}
\begin{proof}
    Let $X$ be an MCS such that $X[i, i+1] = 11$ \changed{for some $i$.} 
    For any $S\in \mathcal S(\mathcal H)$,
    there is an arrangement $(\delta_1, \ldots, \delta_{\size{X}})$ 
    such that $S[\delta_j] = X[j]$ for all $1 \le j \le \size{X}$.
    By construction, in each string of $\mathcal S(\mathcal H)$, there is at least one occurrence of $0$ between any two occurrences of $1$.
    Therefore, $S[\delta_i, \delta_{i+1}]$ contains $101$ as a subsequence, which contradicts the maximality.
\end{proof}

\changed{Note that any MCS must also be a subsequence of $S_0 = (01)^{\size{V}}$, so it can be obtained by deleting some characters from $S_0$. Specifically, by \Cref{obs:11}, when we delete a $0$ we must also delete one of its adjacent $1$s; otherwise such subsequence would not be maximal, as it would contain $11$ as a substring.}

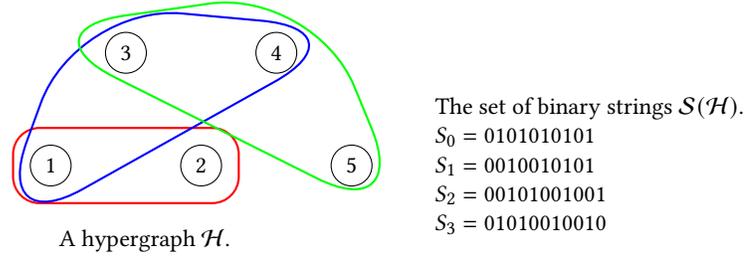
\begin{figure}[t]
    \centering
    \begin{tikzpicture}
        \node[draw, circle] (1) at (0, 0) {1};
        \node[draw, circle] (2) at (2, 0) {2};
        \node[draw, circle] (3) at (1, 1.5) {3};
        \node[draw, circle] (4) at (3, 1.5) {4};
        \node[draw, circle] (5) at (4, 0) {5};
    
        \draw[rounded corners=10pt, thick, red] 
            ($(1) + (-0.5,-0.5)$) rectangle ($(2) + (0.5,0.5)$);
    
         \draw[thick, blue, rounded corners=40pt] 
            ($(1) + (-0.8,-1.0)$) -- ($(3) + (-0.65,0.65)$) -- ($(4) + (1.2,0.3)$) -- cycle;

        \draw[thick, green, rounded corners=40pt] 
            ($(3) + (-1.4,0.3)$) -- ($(4) + (0.5,0.85)$) -- ($(5) + (0.8,-0.8)$) -- cycle;

        \node[anchor=west,align=left] at ($(1) + (0, -1.0)$) 
        {A hypergraph $\mathcal H$.};
        \node[anchor=west,align=left] at ($(5) + (1, 0)$) 
        { The set of binary strings $\mathcal S(\mathcal H)$. \\
         $S_0 = 0101010101$  \\ 
         $S_1 = 0010010101$  \\
         $S_2 = 00101001001$ \\
         $S_3 = 01010010010$};
    \end{tikzpicture}
    \caption{An example of our reduction. 
    \changed{Let $E_1 = \set{1, 2}$, $E_2 = \set{1,3,4}$, and $E_3 = \set{3,4,5}$.}
    The maximal independent sets in $\mathcal H$ are $\set{1,3,5}$, $\set{1, 4, 5}$, $\set{2, 3, 4}$,\changed{ $\set{2,4,5}$} and $\set{2, 3, 5}$.
    The MCSs in $S_0$, $S_1$, $S_2$, and $S_3$ are $01010101$, $01001001$, $01000101$, $00101010$, \changed{$00100101$} and $00101001$.}
    \label{fig:reduction}        
\end{figure}


Let $\mathcal X$ be the set of MCSs in $\mathcal S(\mathcal H)$.
The goal of our proof is to show that there is a bijection between the collection of maximal independent sets in $\mathcal H$ and 
$\mathcal X \setminus \set{(01)^{\size{V}-1}}$.
We first show that $(01)^{\size{V}-1}$ is an MCS of $\mathcal S(\mathcal H)$ under our assumption.
\todo{Ale: Kazuhiro, can you address this issue?

Reviewer comment: For Lemma 4, the hypothesis you need is that for every vertex v in the hypergraph, there is some hyperedge that does not contain v. You should state this explicitly.

Our rebuttal: As for the hypothesis of Lemma 4, we will formalize more explicitly the equivalence between "no independent set of size |V| -1" and "no vertex belonging to every hyperedge" which is discussed on page 10, and include the latter as hypotesis in the statement of the lemma. 
}
\todo{Kazuhiro: OK, this revision is complete.}
\begin{lemma}\label{obs:maximal}
    Suppose that $\mathcal H$ has at least two vertices and has no independent set of size $|V|-1$, \changed{that is, no vertex belonging to every hyperedge}.
    Then, the string $(01)^{\size{V}-1}$ is an MCS of $\mathcal S(\mathcal H)$.
\end{lemma}
\begin{proof}
    Since for each $1 \le i \le |\mathcal E|$, $S_i$ is obtained from $(01)^{|V| + |E_i| - 1}$ by replacing $|E_i|$ occurrences of $01$ with one $0$ each,
    it has $(01)^{\size{V}-1}$ as a subsequence.
    By contradiction, let $X$ be a common subsequence of $\mathcal S(\mathcal H)$ 
    that contains $(01)^{\size{V}-1}$ as a proper subsequence, so that $(01)^{\size{V}-1}$ is not maximal.
    Since $X$ is a subsequence of $S_0$, it can be represented by $(01)^\alpha W(01)^\beta$ for some $W$, where $\alpha + \beta = \size{V}-1$ and $W$ is either $0$ or $01$ from \Cref{obs:11}.
    \changed{However, since each $S_i$ contains exactly $|V|-1$ ones, it follows that $W=0$ and $X = (01)^\alpha 0(01)^\beta$.}
    
    By hypothesis $\mathcal H$ has no independent set with cardinality $\size{V}-1$, so
    there is a hyperedge $E_i$ contained in $V \setminus \set{\alpha + 1}$.
    In other words, $E_i$ does not contain vertex $\alpha + 1$.

    We consider a greedy arrangement of $X$ for $S_i$. 
    In this arrangement, the $\alpha$-th $1$ in $X$ is arranged to the $\alpha$-th $1$ in $S_i$. \changed{This is possible by \Cref{obs:11}.}
    Since $E_i$ does not contain $\alpha + 1$, the two characters after the $\alpha$-th $1$ in $S_i$ are $01$.
    However, three characters after the $\alpha$-th $1$ in $X$ is $001$.
    To arrange $001$, we need $0101$ or $01001$ in $S_i$, \changed{hence we consume one more $1$}.
    Therefore, after arranging $001$ to $S_i$, the remaining occurrences of $1$s in $S_i$ are $\size{V} - \alpha - 3$.
    However, the remainder of $X$ contains $\size{V}-\alpha-2$ occurrences of $1$s, so the greedy arrangement is not possible and $X$ cannot be a subsequence of $S_i$, a contradiction. 
\end{proof}

In what follows, we show a bijection between the collection of maximal independent sets in $\mathcal H$ and the set of MCSs without $(01)^{\size{V}-1}$.
Let $E_i = \set{u_1, \ldots, u_h}$ be a hyperedge in $\mathcal H$.
This hyperedge $E_i$ forbids any independent set in $\mathcal H$ from containing all vertices in $E_i$.
In our reduction,
the role of $S_i$ is to forbid a common subsequence in $\mathcal S(\mathcal H)$ containing 
$R(E_i) = R_1R_2\ldots R_{|V|}$ as a subsequence, where $R_j = 01$ if $j \in E_i$, and $R_j = 0$ otherwise.

\begin{lemma}\label{lemma:forbid}
    Let $E_i = \set{u_1, \ldots, u_h}$ be a hyperedge in $\mathcal H$.
    Then, $S_i$ does not contain $R(E_i) = R_1R_2\ldots R_n$ as a subsequence.
\end{lemma}
\begin{proof}
    We prove the lemma by induction on the cardinality of $E$.
    If $E_i$ consists of only one vertex $u_1$, 
    $R(E_i) = 0^{u_1 - 1}010^{\size{V} - u_1}$ and
    $S_i = (01)^{u_1-1}0(01)^{\size{V} - u_1}$.
    In this case, $0^{u_1 - 1}0$ is a subsequence of $(01)^{u_1 - 1}0$.
    Moreover, $0^{u_1 - 1}01$ is not a subsequence of $(01)^{u_1 - 1}0$.
    From \cref{obs:subseq}, $R(E_i)$ is a subsequence of $S_i$ if and only if $T' = 10^{\size{V} - u_1}$ is a subsequence of $S' = (01)^{\size{V} - u_1}$. 
    Thus, $R(E_i)$ is not a subsequence of $S_i$.
    
    As an induction step, assume that $E_i$ contains $h > 1$ vertices with $u_1 < u_2 < \dots < u_h$.
    Let $R = R(E_i)$.
    Since $R[1, u_1] = 0^{u_1 - 1}0$, 
    $R[1, u_1]$ is a subsequence of 
    $S_i[1, 2u_1-1] = (01)^{u_1 - 1}0$.
    However, $R[1, u_1+1] = 0^{u_1 - 1}01$ is not a subsequence of $S_i[1, 2u_1-1]$.
    Therefore, from \Cref{obs:subseq},
    $R$ is a subsequence of $S_i$ if and only if $R[u_1+1, \size{R}]$ is a subsequence of $S_i[2u_1, \size{S_i}]$.
    For a set of vertices $E' = \set{u_2 - u_1, u_3 - u_1, \ldots, u_h - u_1}$, 
    we obtain the same strings $R[u_1+1, \size{R}]$ and $S_i[2u_1, \size{S_i}]$.
    Since $\size{E'} < h$, $S_i[2u_1, \size{S_i}]$ does not contain $R[u_1+1, \size{R}]$ as a subsequence.
\end{proof}

From \Cref{obs:11} and \Cref{lemma:forbid}, any MCS different from $(01)^{\size{V}-1}$ contains $\size{V}$ \changed{occurrences of} $0$ and
does not contain $R(E)$ as a subsequence for each $E \in \mathcal E$.
We are ready to prove that there is a bijection between the \changed{collection} of maximal independent sets \changed{in $\mathcal H$} and the set of MCSs \changed{of $\mathcal S(\mathcal H)$} different from $(01)^{\size{V}-1}$:
\begin{lemma}\label{lem:cs:is}  
    Let $\mathcal X$ be the set of MCSs of $\mathcal S(\mathcal H)$. Then, there is a bijection $\psi$ \changed{from the collection of maximal independent sets in $\mathcal H$ to $\mathcal X \setminus \set{(01)^{\size{V}-1}}$}.
\end{lemma}
\begin{proof}
    We define a function $\psi$ on a larger domain, the \changed{collection} of \changed{all} independent sets in $\mathcal H$, \changed{and then prove that its restriction of $\psi$ on the collection of maximal independent sets in $\mathcal H$ is an intended bijection.}
    For an independent set $U$, 
    $\psi(U)$ is defined as $X_1\ldots X_{\size{V}}$, where $X_i = 01$ if $i \in U$, 
    otherwise, $X_i = 0$.
    We show that $\psi(U)$ is a common subsequence of $\mathcal S(\mathcal H)$.
    From the definition of $\psi(U)$, $\psi(U)$ is a subsequence of $S_0$.
    Thus, we show that $\psi(U)$ is contained in $S_i$ \changed{as a subsequence} for $1 \le i \le \size{\mathcal E}$.
    
    Since $U$ is an independent set \changed{in $\mathcal H$, we have} $E_i \setminus U \neq \emptyset$.
    Let $h$ be the minimum integer in $E_i \setminus U$.
    From the construction of $S_i$, 
    $S_i$ contains $Y = (01)^{h-1}0(01)^{\size{V} - h}$.
    We decompose $Y$ into $Y_1 \ldots, Y_{\size{V}}$, where $Y_j = 01$ \changed{for each $j$ with} $j \neq h$ and $Y_h = 0$.
    For each \changed{$1\le j \le \size{V}$, $X_j$ is a subsequence of $Y_j$.} 
    Therefore, $\psi(U)$ is a subsequence of $S_i$.
    We next show that if $U$ is maximal, $\psi(U)$ is also maximal. 
    \changed{Suppose by contradiction that $\psi(U)$ is non-maximal. This means that there is a string $T$ that is a supersequence of $\psi(U)$ and a subsequence of $S_0$. Since $\psi(U)$ contains $\size{V}$ occurrences of $0$, $T$ must be obtained by adding some $1$s to $\psi(U)$.
    Without loss of generality, let $T$ be the string that is obtained by adding only one $1$ to $\psi(U)$, between $X_j$ and $X_{j+1}$, say, for a given $j$ such that $X_j=0$. By construction, $\psi(U\cup \set{j}) = T$. By maximality of $U$, there exists a hyperedge $E \subseteq U \cup \set{j}$. Then, $T$ contains $R(E)$ as a subsequence.}
    From \Cref{lemma:forbid}, $T$ is not a common subsequence of $\mathcal S(\mathcal H)$.
    
    \changed{Conversely,} let $X$ be an MCS in $\mathcal S(\mathcal H)$ \changed{different from $(01)^{\size{V}-1}$.} 
    \changed{By \Cref{obs:11}, $X$ does not contain $11$ as a substring, so it has to contain $\size{V}$ occurrences of $0$, otherwise it would be a subsequence of $(01)^{\size{V}-1}$.}
    Therefore, $X$ can be decomposed by $X_1 \ldots X_{\size{V}}$, where each $X_i$ is either $0$ or $01$.
    We define a set of vertices $\psi^{-1}(X)$ as follows:
    $X_i = 01$ if and only if $\psi^{-1}(X)$ contains $i$.
    Let $E_i = \set{u_1, \ldots, u_h}$ be a hyperedge in $\mathcal E$.
    By \Cref{lemma:forbid} \changed{and by transitivity of the subsequence relation}, $X$ does not contain $R(E_i)$ as a subsequence.
    Therefore, $\psi^{-1}(X)$ does not contain $E_i$ and $\psi^{-1}(X)$ is an independent set of $\mathcal H$.
    Finally, suppose by contradiction $\psi^{-1}(X)$ is not a maximal independent set. 
    Then, by construction, the string $\psi(\psi^{-1}(X))$ is a supersequence of $X$ and $\psi(\psi^{-1}(X))$ a common subsequence of $\mathcal S(\mathcal H)$. 
    But this contradicts the maximality of $X$.
\end{proof}


The reduction presented in this section, together with the previous lemma, concludes the proof of \Cref{thm:mis-reduction}:
\misReduction*

This bijection implies that
if there is an output-polynomial time algorithm for
\MCS{} for binary strings, then an output-polynomial time algorithm for \MIS{} and the following corollary holds.
\hardnessMCSBin*

This result implies that output-polynomial time enumeration of \MCS{} for binary strings is \changed{challenging}, since output-polynomial time enumeration of \MIS{} is a long-standing open problem.
Moreover, since we give a bijection between the \changed{collection} of maximal independent sets and the set of MCSs in binary strings, 
counting MCSs is at least as hard as counting the maximal independent sets in hypergraphs: this latter problem is in fact \#\P-complete not just for hypergraphs, but even in the special case of simple graphs, which can be modeled as hypergraphs where all hyperedges have size 2 \cite{doi:10.1137/S0097539797321602}. Thus, we obtain the following result as well:
\hardnessMCScount*

This also immediately implies that there is no assessment algorithm running in time polynomial in the input size, unless $\P = \NP$: using such an assessment algorithm in a binary search manner on $z$, we could obtain the exact number of MCSs with at most $\size{V}$ calls to the algorithm, as the number of maximal independent sets is at most $2^{\size{V}}$. 
On the other hand, an efficient assessment algorithm, in the sense of running time polynomial in the input and in $z$, is an open problem, both for general $z$ and for small (polynomial) values of $z$.



\section{Positive results for parameterized cases} 
\label{section:parameterized}

While we have shown \MCS~to be hard in general, it is worth investigating what parameters can make the problem tractable.
In this section, we present some positive results concerning the parameterized complexity of \MCS{} with respect to the string length~$n$, and the number of input strings $k$.

To look at a related problem, we can observe how the longest common subsequence between $k$ strings of length $n$ cannot be found in $O(n^{k-\epsilon})$, for any $\epsilon>0$, unless the Strong Exponential Time Hypothesis (SETH) is false \cite{AbboudLCShardness}, but this bound becomes tractable for constant values of $k$. Do MCS present a similar situation? 


We note that all the positive results presented here are also valid for the \MCSAS{}, because of its relationship with \MCS{} described in \Cref{section:contribution}. Furthermore, since by enumerating we are also implicitly counting, we also derive the positive results shown in Table~\ref{tab:neg:results} for the \textsc{MCS Counting} and \textsc{MCS Assessment} problems. 







In the following, we will consider how the problem complexity behaves if we consider parts of the input size as parameters.
To correctly classify the complexity, we introduce two complexity classes for parametrized problems.
We say that a problem is \emph{\FPT} with respect to a parameter $p$ if it can be solved in $O(f(p) |\mathcal{J}|^{c})$ time for some constant $c$ and computable functions~$f$~\cite{downey1997parameterized}, where $|\mathcal{J}|$ is the non-parametrized size of the input.
Moreover, we say that a parameterized problem
is in \emph{\XP} with respect to a parameter $p$ if it can be solved in $O(f(p) |\mathcal{J}|^{g(p)})$ time for some computable functions $f,g$ \cite{downey1997parameterized}. 
Clearly, all problems that are \FPT~are also in \XP.
For an enumeration algorithm, we say it is \emph{\XP-delay} for parameter $p$ if its delay is bounded by $O(f(p) |\mathcal{J}|^{g(p)})$ for some computable functions $f,g$.

Before presenting the results, we remark that we can check whether a given string is an MCS of a set $\mathcal{S}$ of strings in $O(k\isize\log(\isize))$ time by using the algorithm of~\cite{hirota2023fast}: indeed, what they denote as ``total length of the strings'' corresponds to what we denote $\isize$, and their $m$ is our $k$. 

\subsection{\MCS{} parameterized by minimum string length}
\label{section:param-string-len}
Consider a set of $k$ strings $\mathcal{S} = \{S_1,\ldots,S_k\}$, and let $\ml$ be the \textit{minimum} of their lengths (clearly, $\ml \le n$). Without loss of generality, assume $\ml$ is the length of $S_1$. 
We show here how MCS enumeration can be efficiently parameterized in $\ml$.

Firstly, observe that any common subsequence $X$ of $\mathcal S$ is a subsequence of the shortest string $S_1$. Consequently, we can encode each common subsequence of $\mathcal{S}$ (even non-maximal ones) as a binary string of length $\ml$, through one of its occurrences in $S_1$: it suffices to indicate for each position of $S_1$ whether we take, or discard, the corresponding character. 

The number of such encodings, and consequently the number of CS and MCS of $\mathcal{S}$, is then surely bounded by $2^{|S_1|} = 2^{\ml}$. Since we can check if a given string is an MCS of $\mathcal{S}$ in 
$O(k\isize\log(\isize))$ time, we can naively enumerate all MCSs by iterating over all possible subsequences of $S_1$, and discarding those that do not pass the maximality check, in 
$O(k\isize \log(\isize)\cdot 2^{\ml})$ time, which is significant for small values of $\ml$. 

Indeed when $\ml \le c\log(\isize)$, for some constant $c$, the bound reduces to a polynomial-time algorithm, running in 
$O(k\isize^{1+c} \log(\isize))$ time.
Analogously, when 
$\ml \le c\log^d(\isize)$ for some constants $c$ and $d$, we obtain a quasi-polynomial-time algorithm, running in $O\left(k\isize \log(\isize)2^{c\log^d(\isize)}\right)$ 
time. We can thus derive the following.

\paramMinLen*

Finally, we note that the value of parameter \ml can sometimes be reduced, obtaining more tractable instances, due to the following observation:
\begin{observation}
    A character occurring in a common subsequence occurs in every input string. Therefore, a character not occurring in every input string can be deleted from the input without altering the set of MCSs. 
\end{observation}
It follows that we can preprocess any instance $\mathcal{S} = \{S_1,\ldots,S_k\}$ by removing every occurrence of such characters, possibly reducing the minimum string length $\ml$. We can also observe that, in this preprocessed instance, $|\Sigma| \le \ml$ because any remaining character must occur in the shortest string.

\subsection{\MCS{} parameterized by the number of strings}
\label{section:param-num-strings}
Consider now the other potential parameter: the number $k$ of strings in the input instance $\mathcal{S} = \{S_1, \ldots, S_k\}$, which we here assume to be constant. We show that the enumeration algorithm in \cite{DBLP:journals/algorithmica/ConteGPU22} can be generalized to handle any number $k$ of strings of maximum length $n$ in $O\left(kn^{k+1}\left ( n^k + kn\log(kn)\right )\right)$ delay. 
In this section, we give a brief sketch of the generalization. 

Recall that we can check for maximality of a given subsequence in $O(k\isize\log(\isize))$ time. 
The algorithm in \cite{DBLP:journals/algorithmica/ConteGPU22} uses two main concepts that need to be generalized to $k$ strings: the set of unshiftables $\mathcal{U}$, given by pairs of matching positions in the two strings that satisfy further specific properties, and, for every valid prefix $P$ of an MCS, the set of valid candidates for extension $Ext_P\subseteq \mathcal U$. 

The former can be easily generalized \changed{to $k$ strings, where each unshiftable becomes a $k$-tuple of matching positions with some constraints, hence $|\mathcal U| \le n^k$.} \changed{In essence, a $k$-tuple $u = (u_1, \ldots, u_k)$, corresponding to character $c(u):=S_i[u_i]$, is unshiftable if there is another unshiftable $v = (v_1, \ldots, v_k)$ (or, as a base case, $v=(|S_1|+1,\ldots,|S_k|+1)$) such that, for each $i$, $u_i$ is the rightmost occurrence of $c(u)$ in $S[1,v_i-1]$ (i.e. before $v_i$ in $S_i$).}
We can trivially extend the characterization of \cite[Fact 1]{DBLP:journals/algorithmica/ConteGPU22} to give a constructive recursive definition of $\mathcal U$, 
running in $O(k|\Sigma|n^k)$. 

Secondly, consider $Ext_P$: for a given string $P$ that is known to be a prefix of an MCS, the set $Ext_P$ is a set of unshiftables corresponding to candidate extensions of $P$. Being able to compute $Ext_P$ (and refine it with maximality checks), immediately gives us a backtracking enumeration algorithm for MCSs, as we can find every $c$ such that $P\cdot c$ is still a prefix of an MCS. Again, we can extend its definition of $Ext_P$ from \cite{DBLP:journals/algorithmica/ConteGPU22} to the case of $k$ strings:
\cite[Definitions 7-8]{DBLP:journals/algorithmica/ConteGPU22} define a set of unshiftable $2$-tuples immediately following $P$ that are ``not dominated'' (i.e., no other unshiftable occurs completely in between them and $P$), and this can be trivially considered with $k$-tuples on $k$ strings instead of $2$. Even if some of the more refined optimizations in \cite{DBLP:journals/algorithmica/ConteGPU22} do not immediately extend to the case of $k$ strings, the $Ext_P$ set can still be naively computed in $O(kn^{2k})$ time by pairwise checking all the elements in $\mathcal U$, and eliminating dominated ones. 

\begin{algorithm}[t]
\caption{\textbf{ (\XP-Delay Enumeration Algorithm for MCS) }}
\label{alg:mcs-enum-kstrings}
\KwIn{\emph{$\mathcal{S} = \{S_1, \ldots, S_k\}$ }} 
\KwOut{\emph{$MCS(\mathcal{S})$}}
\SetAlgoLined\DontPrintSemicolon
  \SetKwFunction{enum}{EnumerateMCS}
  \SetKwFunction{unshift}{FindUnshiftables}
\SetKwFunction{bpart}{BinaryPartition}

\SetKwProg{myalg}{Algorithm}{}{}
\myalg{\enum{$\mathcal{S}$}}{
        Let $\#,\$\not\in \Sigma$; add them resp. at positions $0$ and $|S_i|+1$ of each $S_i$\;
        $\mathcal U$ =  \unshift{$(|S_1|+1,...,|S_k|+1)$}{}\;
        \bpart{$\#$, $(0,...,0)$}\;
    }
    \vspace{4mm}
  \KwIn{\emph{a valid prefix $P$, the indices of the shortest prefixes of $S_1,...,S_k$ containing $P$}} 
\KwOut{\emph{all strings in $MCS(\mathcal{S})$ having $P$ as prefix}}
  \SetKwProg{myproc}{Procedure}{}{}
  \myproc{\bpart{$P$, $(\ell_1,...,\ell_k)$}}{
        Let $\mathcal U_P = \{(u_1,...,u_k)\in \mathcal{U} \ |\ u_i > \ell_i \ \forall i \}$ \;
        Compute $Ext_P$ \label{ln:findext} by eliminating dominated elements from $\mathcal U_P$ \;
        \eIf{$Ext_P = \emptyset$} {\textbf{output} $P$}{
            \For{each $c \in \Sigma $ corresponding to some $(u_1,...,u_k)\in Ext_P$}{
                \If {$P\in MCS(S_1[0,u_1],..., S_k[0,u_k])$} {
                For each $j$, let $i_j$ be the first occurrence of $c$ after $\ell_j$ in $S_j$\;
                \bpart{$P\, c$, $(i_1,...,i_k)$}\;
                }
            }
        }
  }
\vspace{4mm}
 \KwIn{\emph{An unshiftable 
 $(i_1,...,i_k)$ of strings $S_1,...,S_k$}} 
\KwOut{\emph{The set of unshiftable $k$-tuples of strings $S_1[1,i_1-1],..., S_k[1,i_k-1]$}}
\SetKwProg{myunshift}{Function}{}{}
\myunshift{\unshift{$(i_1,...,i_k)$}}{
    Let $\mathcal U = \emptyset$ \;
    \For{$c \in \Sigma$}{
        For each $j$, let $r_j$ be the rightmost occurrence of $c$ in $S_j[1,i_j-1]$ \; 
        \If {$r_j\not =-1$ for all $j$, and $(r_1,...,r_k) \not \in \mathcal U$}{
            Add $(r_1,...,r_k)$ to $\mathcal{U}$\;
            $\mathcal U = \mathcal U \ \cup $ \unshift{$(r_1,...,r_k)$}
        }
    }
    \Return $\mathcal{U}$
  }

\end{algorithm}

With these definitions, the characterization given by \cite[Theorem 3]{DBLP:journals/algorithmica/ConteGPU22} generalizes as well, yielding the following binary partition enumeration algorithm (see Algorithm~\ref{alg:mcs-enum-kstrings}): 
first, we compute the set $\mathcal U$ in $O(k|\Sigma|^2 n^k)$ time. Then, in each binary partition recursive call, associated with a given valid prefix $P$, we (1) check if $P$ is an MCS, and in that case output $P$; otherwise (2) we compute the $Ext_P$ set in $O(kn^{2k})$ time, and then check for each of the $O(n^k)$ elements $u = (u_1,\ldots,u_k) \in Ext_P$ whether $P$ is an MCS of $\{S_1[1,u_1-1],\ldots,S_k[1,u_k-1]\}$. If this is the case, we recurse on prefix $P\cdot c(u)$. Therefore, since the binary partition tree depth is $O(n)$, the delay of such an algorithm is bounded by $O\left (n \left (k\isize\log(\isize) + kn^{2k} +  kn^k\isize\log(\isize) \right )\right )=O\left (kn^{2k+1} + kn^{k+1}\isize\log(\isize)\right)$ time, and since $k\ge 2$ the factor $\isize\log(\isize) \le kn\log(kn)$ is subsumed, as well as the $O(k|\Sigma|^2 n^k)$ preprocessing time. We can thus conclude:

%
%
%

\paramk*

It is worth remarking that the space requirement is $O(n^k)$ due to storing the unshiftables, and that the above complexity can be given in the standard RAM model because of our assumption on $k$: the natural choice of word size is $O(k \log n)$ bits to index the unshiftables data structure, and the word size should be logarithmic in the input size, which is $\isize = O(nk)$. We observe how, since $k$ is a constant, this requirement is met. However, further efforts to develop an enumeration algorithm for non-constant values of $k$ should take this issue into account.

\section{Conclusion}

In this paper we studied the complexity of enumerating Maximal Common Subsequences (MCS), a special case of Maximal Frequent Subsequence Mining, and some related problems such as counting and assessing the number of MCSs.
We showed that, in general, \MCS{} is intractable, as no algorithm can enumerate all MCSs in polynomial time with respect to the combined size of the input and the output, unless $\P = \NP$. 
This result is based on a reduction from 3-\SAT{} to a decision problem that, given a set of input strings and a set $\mathcal{Z}$ of MCSs, asks whether there is another MCS not included in $\mathcal{Z}$.

We also show that no output-polynomial algorithms exist for constructing polynomial-sized indexes that store all MCSs and allow to efficiently query any specific solution. Furthermore, we prove that MCS assessment, where we are asked to decide whether a given instance has at least $z$ solutions, is \NP-hard even when $z$ is linear in the input size. 

In light of these negative results, we explore potential parametrizations for the studied problems. 
In the case of binary strings, we construct a one-to-one reduction between the set of Maximal Independent Sets (MIS) in hypergraphs and the set of MCSs of a particular instance.
This allows us to conclude that counting MCSs is \#\P-Complete, and
that \MCS{} is as hard as \MIS{} on hypergraphs, for which the existence of an efficient algorithm is a long-standing open problem.
On the positive side, for bounded string lengths we identify an \FPT~algorithm for solving all the problems. When the number of strings is instead bounded, we provide an \XP-delay algorithm for \MCS{}. 

However, it remains open whether efficient algorithms exist for \MCSC{} or \textsc{MCS indexing} when the number of strings is treated as a parameter, as well as whether 
we can perform efficient assessment (for $z$ polynomial in the input size) in the case of binary strings.

\section*{Acknowledgments}
We wish to thank Roberto Grossi and Rossano Venturini for useful conversations.
This work is partially supported by JSPS KAKENHI Grant Numbers 
JP21K17812, 
JP22H03549, 
JP23K28034, 
JP24H00686, 
JP24H00697, and 
JP25K21273, 
and JST ACT-X Grant Number JPMJAX2105. 
Author GP is supported by the Italian Ministry of Research, under the complementary 36 actions to the NRRP ``Fit4MedRob - Fit for Medical Robotics'' Grant (\#PNC0000007).
Authors GB and AC are partially supported by MUR PRIN 2022 project EXPAND: scalable algorithms for EXPloratory Analyses of heterogeneous and dynamic Networked Data (\#2022TS4Y3N). 

\bibliographystyle{abbrv}
\bibliography{main.bib}

\end{document}